\documentclass[12pt,twoside,a4paper,leqno]{article}
\usepackage{graphicx}
\graphicspath{{./}}
\usepackage{amssymb}
\usepackage{amsthm}
\usepackage{amsmath}
\usepackage{fleqn}

\def \v#1{\relax
          \ifmmode{\bf #1\/}\else{$\bf #1\/$}\fi}
\def \M#1{\relax
          \ifmmode{#1}\else{$#1$}\fi}
\def \affr #1{\relax
              \ifmmode{\cal #1}\else{$\cal #1$}\fi}
\def \koord #1{\relax
              \ifmmode{\hbox{$\sf #1$}}\else{$\sf #1$}\fi}

\newcommand{\I}{\hat{\imath}}

\begin{document}

\begin{center}
\LARGE{Analyzing Midpoint Subdivision}
\\[1.5em]
\large{Hartmut Prautzsch\footnote{E-mail address: prautzsch@kit.edu} and Qi Chen\footnote{Corresponding author, E-mail address: qi.chen@kit.edu}}
\\[0.5em]
\large{Karlsruher Institut f\"ur Technologie (KIT)\\
Germany 
}
\end{center}

\begin{abstract}
Midpoint subdivision generalizes the Lane-Riesenfeld algorithm for uniform tensor product splines and can also be applied to non regular meshes. For example, midpoint subdivision of degree $2$ is a specific Doo-Sabin algorithm and midpoint subdivision of degree $3$ is a specific Catmull-Clark algorithm. In 2001, Zorin and Schr\"oder were able to prove $C^1$-continuity for midpoint subdivision surfaces analytically up to degree $9$. Here, we develop general analysis tools to show that the limiting surfaces under midpoint subdivision of any degree $\ge 2$ are $C^1$-continuous at their extraordinary points.
\end{abstract}

\begin{flushleft}
\end{flushleft}

\swapnumbers

\theoremstyle{plain}
\newtheorem{DEF}[equation]{Definition}
\newtheorem{SATZ}[equation]{Theorem}
\newtheorem{LEMMA}[equation]{Lemma}
\newtheorem{FOLGERUNG}[equation]{Corollary}
\newtheorem{FIGCaption}[equation]{Figure}
\makeatletter
\def\tagform@#1{\maketag@@@{\ignorespaces#1\unskip\@@italiccorr}}
\makeatother

\numberwithin{equation}{section}
\renewcommand{\theequation}{(\thesection.\arabic{equation})}

\newcommand{\FigureCenter}[4]{
    \begin{figure}[#1]
        \begin{center}
            #3
        \end{center}
        \begin{FIGCaption}\label{#2}
            {#4} 
        \end{FIGCaption}
    \end{figure}
}

\tableofcontents

\section{Introduction}
The midpoint subdivision schemes form a class of subdivision schemes for arbitrary two-manifold meshes. The midpoint subdivision scheme of degree $n \in \mathbb N$ is given by the operator
\[\M M_n = \M A^{n-1} \M R\;,\]
where $\M R$ and $\M A$ are the refinement and averaging operators, respectively. Refining a mesh $\affr N$ by $\M R$ means to connect the center of each face of $\affr N$ with the midpoints of all its edges, which results in the quadrilateral mesh $\M R\affr N$. Averaging $\affr N$ means to connect the centers of all adjacent faces, which results in the mesh $\M A\affr N$, see Figure~\ref{Abb:M2}. 
\FigureCenter{h}{Abb:M2}
  {\includegraphics{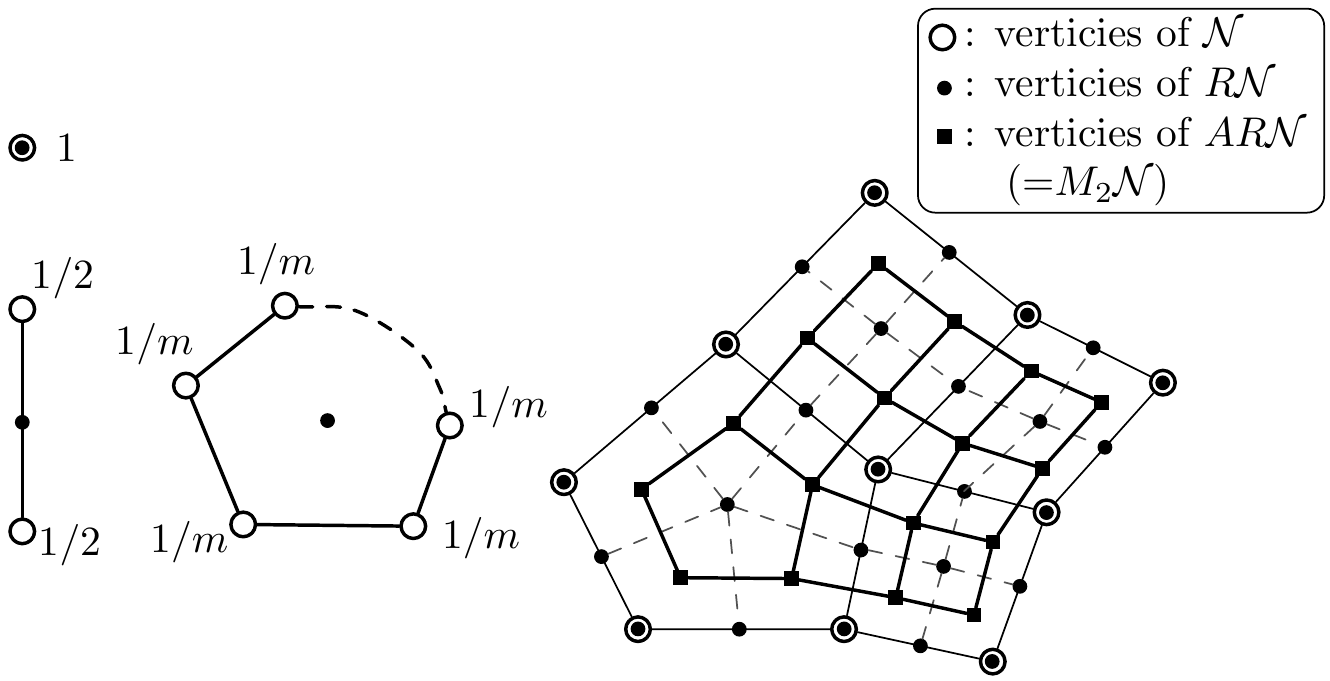}}
  {Masks of $\M R$ (left and middle) and $\M A$ (middle) and an example for $\M M_2 \affr N$ (right).}

The midpoint subdivision schemes of degree $2$, $3$ and $4$ are specific instances of the well-known Doo-Sabin, Catmull-Clark and Qu algorithms \cite{DS78, CC78, Qu90}.

The \emph{valence} $m$ of a vertex or a face is the number of incident or surrounding edges. An inner vertex or face is called \emph{regular} if $m = 4$ and \emph{irregular} or \emph{extraordinary} otherwise.
After subdividing a mesh by $\M M_n$, the mesh has no irregular vertices or faces if $n$ is even or odd, respectively. Subdividing the mesh further does not increase the number of extraordinary vertices and faces and extraordinary elements become more and more isolated. 

Therefore, and since midpoint subdivision generalizes the Lane-Riesenfeld algorithm, midpoint subdivision surfaces are spline surfaces except for finitely many extraordinary points. At these points, the smoothness analysis is complicated. Using Reif's $C^1$-criterion \cite[Theorem 3.6]{Reif95} and interval arithmetic, Zorin and Schr\"oder \cite{ZS2001} showed \mbox{$C^1$-smoothness} for degrees $n=2, \ldots, 9$.

Naturally, this numeric approach is limited to a finite number of degrees. 
Here, we develop a geometric framework that allows us to prove $C^1$-continuity of midpoint subdivision surfaces of all degrees. Moreover, we think that this framework may conceptually be useful for other classes of (simple) subdivision schemes and that it provides -- for the first time -- an analysis toolbox for a complete class of subdivision schemes.

In Section~\ref{SECTION:Rings}, we discuss the basic topological dependencies between the vertices of a mesh before and after a subdivision step. These dependencies correspond to a certain block structure of the subdivision matrix with a block possessing a strictly positive power. This particular block represents the subdivision operator restricted to a certain central part of the mesh, which we call a \emph{core mesh}. In Section~\ref{SECTION:SymmetricRingnets}, definitions and properties of symmetric meshes and symmetric subdivision schemes are introduced.  
For symmetric meshes, in Section~\ref{SECTION:ComparingRingnets}, we define a partial order based on particular coordinate systems. In Section~\ref{SECTION:KonvergenzDerFolge1}, we show that subdividing specific symmetric grid-like core meshes results, in the limit, under normalization in a symmetric mesh whose regular vertices do not coincide with the center and show that it is an eigenvector of the subdivision matrix restricted to core meshes. In Section~\ref{SECTION:Subdominance}, we use the partial ordering introduced in Section~\ref{SECTION:ComparingRingnets} to compare the eigenvector constructed in Section~\ref{SECTION:KonvergenzDerFolge1} to any other eigenvector and show its subdominance. 
Hence, in Sections~\ref{SECTION:KonvergenzDerFolge1} and \ref{SECTION:Subdominance}, we restrict the analysis to core meshes for which the subdivision matrix has a strictly positive power. Then, in Section~\ref{SECTION:LargerRingnets}, we extend the analysis to larger meshes using the particular block structure of the subdivision matrix. In Section~\ref{SECTION:CharacteristicMap}, we show that subdominant eigenvectors define a regular surjective characteristic map, which concludes the proof that generic midpoint subdivision surfaces of any degree $n(\ge 2)$ are $C^1$-continuous. 
This result is stated in Theorem~\ref{SATZ:C1Mn} and the entire paper consists of its proof which is composed of $19$ lemmata, theorems and corollaries whose ramified interdependencies are depicted in Figure~\ref{Abb:Relationship}.

\FigureCenter{!h}{Abb:Relationship}
  {\includegraphics{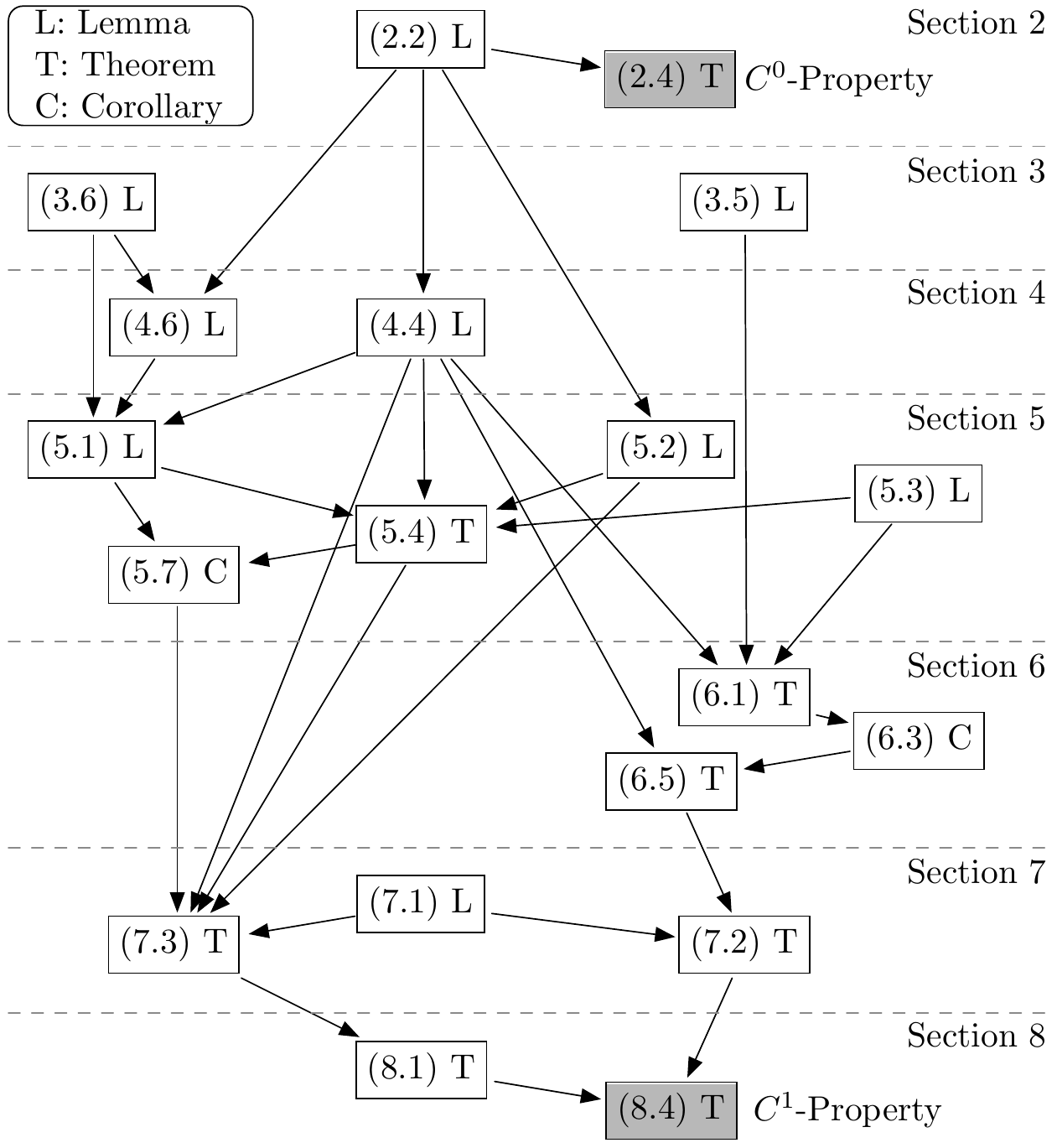}}
  {Interdependencies of lemmata, theorems and corollaries in this paper.}

\section{Rings and ringnets}\label{SECTION:Rings}
For the smoothness analysis of midpoint subdivision surfaces at extraordinary points, it is sufficient for odd $n$ to consider a mesh with only one irregular vertex and for even $n$ to consider a mesh with only one irregular face. We will assume this throughout the paper. These simple meshes are illustrated in Figure~\ref{Abb:BSP_Ring_Netz} and are called \emph{ringnets}.

\FigureCenter{h}{Abb:BSP_Ring_Netz}
  {\includegraphics{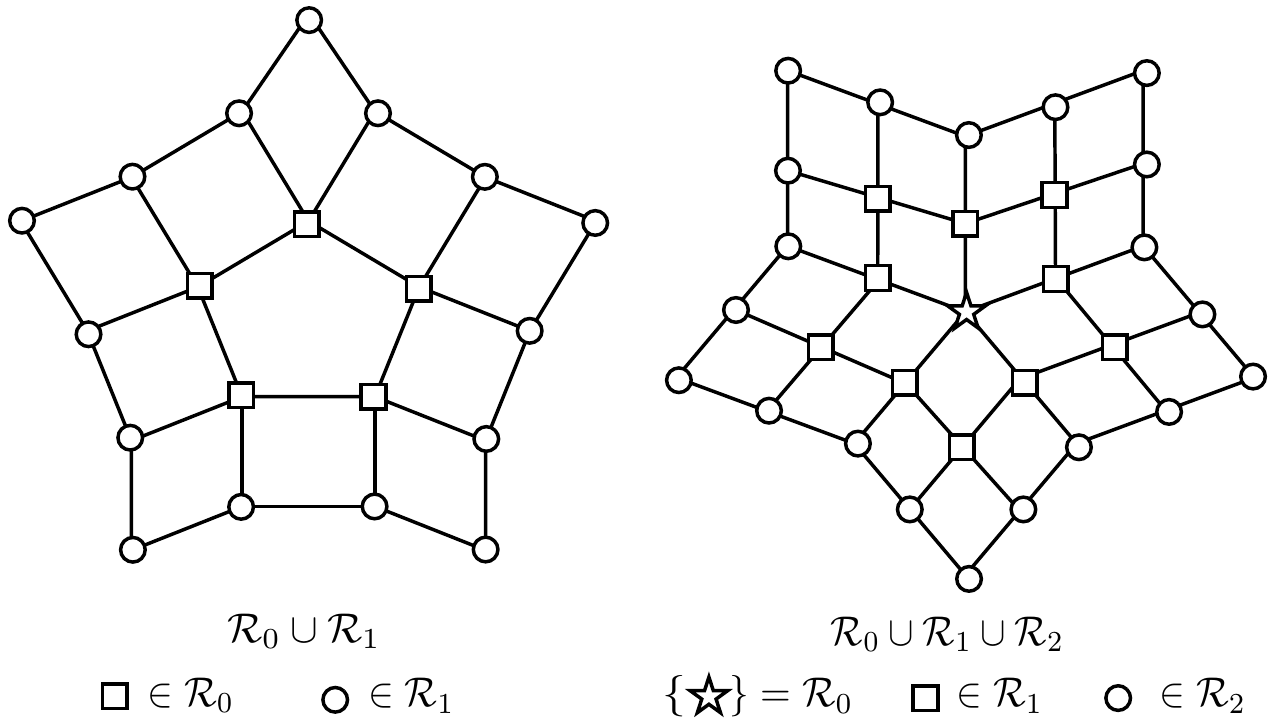}}
  {Example of rings and ringnets: a $1$-ringnet with an irregular face of valence $5$ (left) and  a $2$-ringnet with an extraordinary vertex of valence $5$ (right).}

The $0$-ring of a ringnet is formed by the irregular vertex or the irregular face, the $1$-ring by the vertices of the adjoining faces that do not belong to the $0$-ring. The $2$-ring, $3$-ring, etc. are defined similarly. The rings $0,1,\ldots,l$ constitute an $l$-\emph{ringnet} or an $l$-\emph{net} for short. 
A ringnet $\affr N $ is understood as a vector whose coordinates represent the vertices of the ringnet or as the set of its vertices, depending on the context.
Following other authors, we call a ringnet with an extraordinary vertex \emph{primal} and a ringnet with an extraordinary face \emph{dual}.

In this paper, $n$ always denotes the degree of the midpoint subdivision scheme and $m$ the valence of a vertex or of a face. 
In the following, let $\affr N$ be a sufficiently large ringnet and let $\affr N^{(k)}$ be the subdivided ringnet $\M M_n^k \affr N$. If $\affr N$ is primal we require $n$ to be odd and otherwise to be even.  
The $i$-th ring is denoted by $\affr N_i$, and the subnet built from $\affr N_i, \affr N_{i+1}, \ldots, \affr N_j$ is denoted by $\affr N_{i.. j}$. 
In particular, $\affr N_{0.. i}$ is called an \emph{$i$-ringnet} or \emph{$i$-net} or \emph{$i$-mesh} for short.  

In the rest of this section, we study on which vertices of $\affr N$ the vertices of $\affr N':=\affr N^{(1)}$ depend and derive the minimal size of a ringnet required to examine the eigenvalues of the subdivision matrix.

\begin{LEMMA}\label{Lemma:BeeinflussDerRingeNachEinerUnterteilung}
\textbf{\emph{(The influence of vertices under subdivision)}}\\
We say that $\affr N_i$ \emph{influences}  $\affr N_j'$ and denote this by $\affr N_i \stackrel{\M M_n}{\rightsquigarrow} \affr N_j'$ if every vertex in $\affr N_i$ influences some vertex in $\affr N_j'$ and if additionally all vertices in $\affr N_j'$ depend on $\affr N_i$. This is the case if and only if
\[2i-\left\lfloor \frac{n+1}{2}\right\rfloor \leq j \leq 2i + \left\lceil \frac{n+1}{2}\right\rceil\;.\]
Consequently, $\affr N_{0..j}'$ is completely determined by $\affr N_{0..j}$, i.\,e.,
\[\affr N_{0..j}' = (\M M_n \affr N_{0..j})_{0..j}\]
for all $j\ge \omega :=\left\lfloor \frac{n-1}{2}\right\rfloor$. 
Furthermore, for $i\le\omega$, it follows that  
\[\affr N_i^{(0)} \stackrel{\M M_n}{\rightsquigarrow}
 \affr N_{\omega-1}^{(1)} \stackrel{\M M_n}{\rightsquigarrow}
 \affr N_{\omega-3}^{(2)} \stackrel{\M M_n}{\rightsquigarrow}
 \affr N_{\omega-7}^{(3)} \stackrel{\M M_n}{\rightsquigarrow}
 \cdots \stackrel{\M M_n}{\rightsquigarrow}
 \affr N_{0}^{\left(\left\lceil \log_2 (\omega+1) \right\rceil\right)}\]
and since $\affr N_{0}$ or $(\M R\affr N)_{0}$ consists of only one vertex, every vertex in $\affr N_i$ influences every vertex in $\affr N_{0..\omega}^{(k)}$, where $k > \left\lceil \log_2 (\omega+1) \right\rceil$.
\end{LEMMA}
\begin{proof}
{
The lemma can be shown via induction if one observes first that  
\[\affr N_i \stackrel{\M A^2}{\rightsquigarrow} (\M A^2 \affr N)_j\quad \mbox{if and only if}\quad i-1\leq j \leq i+1\;,\]
secondly for a primal $\affr N$ that
\[\affr N_i \stackrel{\M R}{\rightsquigarrow} (\M R \affr N)_j\quad \mbox{if and only if}\quad 2i-1\leq j \leq 2i+1\;,\]
and finally for a dual $\affr N$ that
\[\affr N_i \stackrel{\M A \M R}{\rightsquigarrow} (\M A \M R \affr N)_j\quad \mbox{if and only if}\quad 2i-1\leq j \leq 2i+2\;.\]
}
\end{proof}

The $\omega$-mesh $\affr N_{0..\omega}$ with 
\[\omega = \left\lfloor \frac{n-1}{2}\right\rfloor\]
is called the \emph{core mesh} of $\affr N$ (with respect to $M_n$). According to Lemma~\ref{Lemma:BeeinflussDerRingeNachEinerUnterteilung}, the core mesh of $\affr N$ is the largest sub-ringnet of $\affr N$ with the property that each of its vertices influences every vertex in the core mesh after several iterations of subdivision.

The meshes $\affr N^{(i)} = \M M_n^i \affr N$ converge to a piecewise polynomial surface. Each regular subnet of $(n+1)\times (n+1)$ vertices of $\affr N^{(i)}$ defines a polynomial patch of this surface and all these patches form a spline surface ring $\v s_i$ with an $m$-sided hole. The difference surfaces $\v s_i \backslash \v s_{i-1}$ are smaller spline rings consisting of $3m \left\lfloor \frac{n}{2}\right\rfloor^2$ polynomial patches determined by $\affr N_{0..\rho}^{(i)}$ with
\begin{equation}\label{EQ:RHO}
\rho = \left\lceil \frac{3}{2}n - \frac{3}{2} \right\rceil\;,
\end{equation}
see \cite{Prautzsch98}.
The operator $\M M_n$ restricted to $\rho$-nets is represented by a quadratic matrix $\M S_\rho$, called the \emph{subdivision matrix}.

\begin{SATZ}\label{SATZ:C0Mn}
\textbf{\emph{($C^0$-property of $\M M_n$)}}\\
The subdivision surfaces generated by $\M M_n$ are $C^0$ for all $n\geq 1$.  
\end{SATZ}
\begin{proof}
{
Since the subdivision matrix $\M S_\rho$ is stochastic, $1$ is the dominant eigenvalue of $\M S_\rho$.  
Using Lemma~\ref{Lemma:BeeinflussDerRingeNachEinerUnterteilung}, we obtain that every vertex in $\affr N_{0}$ influences all vertices in $\M S_\rho^2 \, \affr N_{0..\rho}$.
This implies that $\M S_\rho^2$ has a positive column.
According to \cite[Theorem 2.1]{MP89a}, the sequence $(\M S_\rho^i \, \v c)$ converges to a multiple of the vector $[1\,\ldots\, 1]^\mathrm{t}$ as $i \to \infty$ for all real vectors $\v c$.
Therefore, the only dominant eigenvalue of $\M S_\rho$ is $1$ and has algebraic multiplicity $1$.

Hence, the difference surfaces $\v s_i \backslash \v s_{i-1}$ converge to a point, see \cite[Theorem 3.2]{Reif95} and \cite[Remark 35 on Page 35]{Chen2005}, from which it follows that the surfaces generated by $\M M_n$ are continuous.
}
\end{proof}

\section{Symmetric ringnets}\label{SECTION:SymmetricRingnets}
In this section, we introduce grid meshes and show that reflection symmetric eigennets of $\M M_n$ have real eigenvalues.

\begin{DEF}\label{DEF:Netzstruktur}
\textbf{\emph{(Grid mesh)}}\\
A \emph{primal grid mesh} of valence~$m$ and frequency~$f$ is a  planar primal ringnet with the complex vertices  
\[\v g_{ij}^l = i e^{\I 2 \pi lf/m} + j e^{\I 2 \pi (l+1)f/m}, \quad i,j\geq 0,\; l\in \mathbb Z_m,\; \I = \sqrt{-1}\;.\]
A \emph{dual grid mesh} of valence~$m$ and frequency~$f$ consists of the vertices
\[\v h_{ij}^l = \frac{1}{4}(\v g_{i-1,j-1}^{l} + \v g_{i,j-1}^{l} + \v g_{i-1,j}^{l} + \v g_{i,j}^{l}), \quad i,j\geq 1,\; l\in \mathbb Z_m\;,\]
see Figure~\ref{Abb:SterngitterStruktur}.  
For $l$ fixed, the vertices $\v g_{ij}^l$ or $\v h_{ij}^l$ with $(i,j)\neq (0,0)$ of a grid mesh $\affr N$ build the $l$-th \emph{segment} of $\affr N$. The vertices $\v g_{ij}^0$ or $\v h_{ij}^0$ with $i\geq j$ and $(i,j) \neq (0,0)$ constitute the \emph{first half segment} of $\affr N$, denoted by $H(\affr N)$.
We call the vertices $\v g_{0j}^{l} = \v g_{j0}^{l+1}$, $\v g_{ii}^{l}$ and $\v h_{ii}^{l}$ with $i, j \geq 1$ \emph{spoke vertices} and call the vertices $\v g_{ij}^{l}$ or $\v h_{ij}^{l}$ with $i\neq j$ and $i,j\geq 1$ \emph{inner vertices} of the $l$-th segment.
The \emph{segment angle} of  $\affr N$  is $\varphi = 2 \pi f/m$.
The half-line from the center $\v g_{00}^{l}$ through $\v g_{10}^{l}$ or through $\v g_{11}^{l}$ is called the $l$-th \emph{spoke}  or the $(l+\frac{1}{2})$-th \emph{spoke}, denoted by $S_l(\affr N)$ and $S_{l+\frac{1}{2}}(\affr N)$, or $S_l$ and $S_{l+\frac{1}{2}}$ for short, see Figure~\ref{Abb:NetzSpeiche}.
\end{DEF}

\FigureCenter{h}{Abb:SterngitterStruktur}
  {\includegraphics{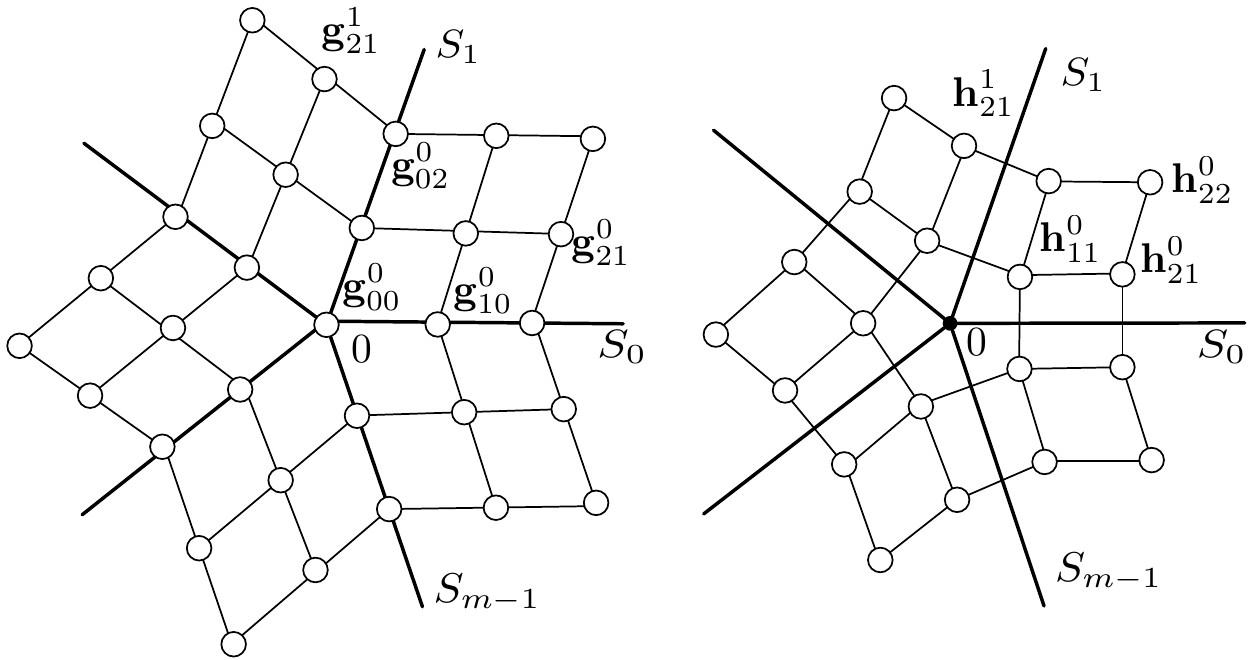}}
  {Primal grid mesh (left) and dual grid mesh (right).}

\FigureCenter{!h}{Abb:NetzSpeiche}
  {\includegraphics{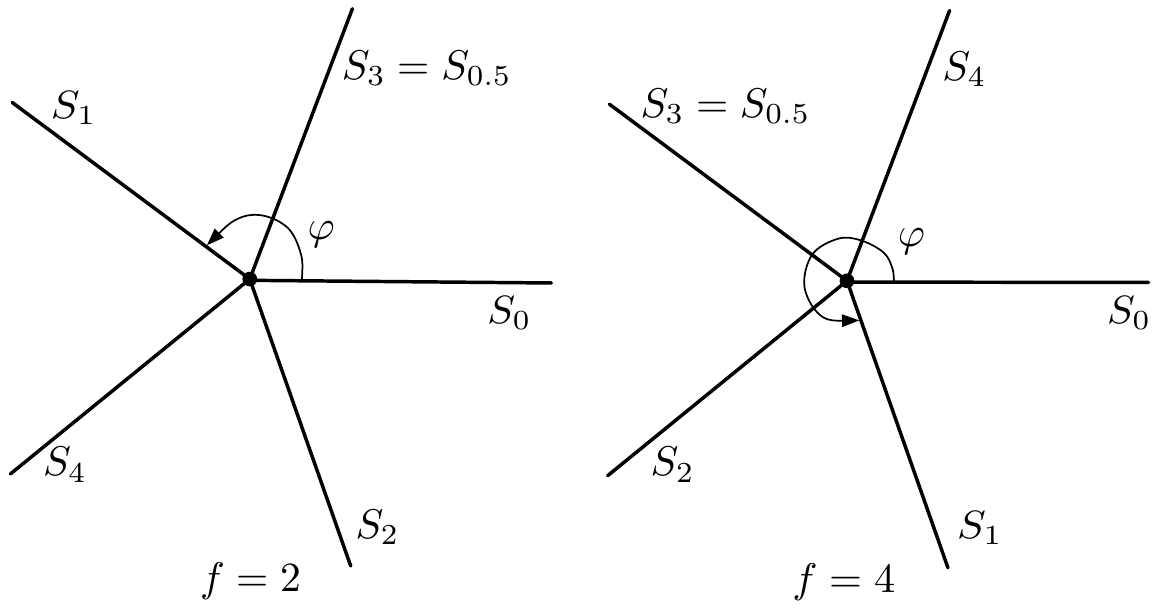}}
  {The spokes of a ringnet with $m=5$ and frequency $f=2$ (left) and $f=4$ (right). The segment angles are $\frac{4}{5}\pi$ (left) and $\frac{8}{5}\pi (\equiv -\frac{2}{5}\pi$) (right).}

Topologically, any ringnet $\affr M$ is equivalent to a grid mesh $\affr N$. Therefore we use the same indices for equivalent vertices and denote the vertices of $\affr M$ by $\v p_{ij}^l$.  

\begin{DEF}
\textbf{\emph{(Symmetric ringnet)}}\\
A ringnet of valence~$m$ with the vertices $\v p_{ij}^l$ in $\mathbb C$ is called \emph{rotation symmetric} with frequency $f$, if
\[\v p_{ij}^l e^{\I 2 \pi f/m} = \v p_{ij}^{l+1}\;.\]
A ringnet $\affr N$ in the complex plane is called \emph{reflection symmetric} if its permutation $\widetilde{\affr N}$ consisting of the points $\widetilde{\v p_{ij}^l} = \v p_{ji}^{(m-1)-l}$ equals the conjugate ringnet $\overline{\affr N}$, i.\,e.,  
\[\widetilde{\affr N} = \overline{\affr N}\;.\]
A rotation and reflection symmetric ringnet is called \emph{symmetric}.
\end{DEF}

It is easy to see that a symmetric ringnet is geometrically symmetric with respect to each spoke. 
Hence, we obtain
\begin{LEMMA}\label{LEMMA:EigenwertVSymmEigennetz}
\textbf{\emph{(Eigenvalue of a reflection symmetric eigennet)}}\\
If $\affr N$ is a reflection symmetric eigennet of $\M M_n$ with eigenvalue $\lambda$ and with $j+1$ rings, $j\ge \omega$, i.\,e., $(\M M_n \affr N)_{0..j} = \lambda \affr N$, then $\lambda$ is real.
\end{LEMMA}
\begin{proof}
{
Since
\[\lambda \affr N = (\M M_n \affr N)_{0..j} = (\M M_n \widetilde{\overline{\affr N}})_{0..j} = (\widetilde{\M M_n \overline{\affr N})_{0..j}} = \widetilde{\overline{\lambda} \; \overline{\affr N}} = \overline{\lambda} \affr N\;,\]
we get $\lambda = \overline{\lambda}$.  
}
\end{proof}

\begin{LEMMA}\label{LEMMA:PosInvariance}
\textbf{\emph{(Positional invariance)}}\\
Let $\affr N$ be a symmetric ringnet with segment angle $\varphi \in (0, \pi]$. If the first half segment $H(\affr N)$ of $\affr N$ lies in the cone $C(\affr N)$ spanned by the spokes $S_0$ and $S_{0.5}$ of $\affr N$, i.\,e.,
\[C(\affr N) = \; \mbox{convex hull of $S_0(\affr N)$ and $S_{0.5}(\affr N)$},\]
then the subdivided ringnets $R \affr N$, $A \affr N$ and consequently $M_n \affr N$ have the same property. By symmetry, the similar statement also holds for all half and, hence, for all full segments.
\end{LEMMA}
\begin{proof}
{
The vertices of $H(R \affr N)$ and $H(A \affr N)$ are convex combinations of $H(\affr N) \cup \{\v 0\}$ or lie on $S_0$ and $S_{0.5}$ since $R$ and $A$ preserve symmetry.
}
\end{proof}

\section{Comparing ringnets}\label{SECTION:ComparingRingnets}
In this section, we partially order symmetric ringnets by comparing the coordinates of all vertices in the first half segment of any two meshes. Since we need that orthogonal projections into the spokes preserve this order, we choose the coordinate axes perpendicular to the spokes. This partial ordering is preserved by midpoint subdivision. Moreover, if $\affr N_1 \gvertneqq \affr N_2$ for two distinct core meshes, sufficiently many subdivision steps lead to the strict inequality $M^k \affr N_1 > M^k\affr N_2$, where $M$ is the operator $\M M_n$ restricted to core meshes, see Lemma~\ref{LEMMA:VergleichbareEigenschaftenZweierSymmetrischenNetzen}. This will be crucial to prove that subdividing a grid mesh leads to a subdominant eigennet and in Lemma~\ref{LEMMA:MIN-MAX-Beziehung}, we show that this eigennet has no zero control points except for its extraordinary point if the net was primal.

The coordinate system $\koord K$ of a planar symmetric ringnet $\affr N$ with segment angle $\varphi \in (0, \; 2\pi)$ has the basis vectors
\[\left[ \begin{array}{c} \cos \theta \\ \sin \theta \end{array} \right] \quad \mbox{and} \quad \left[ \begin{array}{c} 0 \\ 1 \end{array} \right],
\quad \mbox{where} \quad \theta = \varphi/2 - \pi/2\;,\]
which are orthogonal to the spokes $S_{0.5}$ and $S_{0}$, see Figure~\ref{Abb:DKS}.

\FigureCenter{ht}{Abb:DKS}
  {\includegraphics{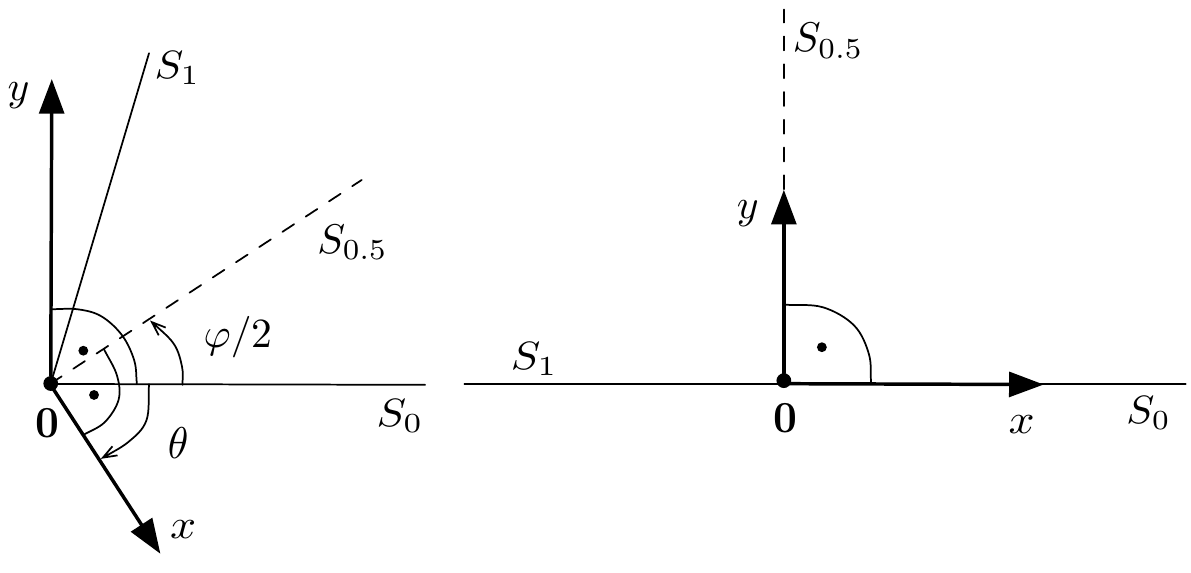}}
  {Examples of $\koord K$ for $\varphi/2 = \measuredangle({S}_0, {S}_{0.5}) \in (0, \frac{\pi}{2})$ (left) and $\varphi/2 = \measuredangle({S}_0, {S}_{0.5}) =\frac{\pi}{2}$ (right).}

The midpoint of two vertices whose positions are symmetric with respect to $S_{0}$ or $S_{0.5}$ is the orthogonal projection of these vertices into $S_{0}$ or $S_{0.5}$, respectively. Thus, the two vertices and their midpoint have the same $x$ or $y$ coordinates in $\koord K$, respectively. This property is crucial in our proofs, see Equation~\ref{EQ:PunkteBez2}, Figures~\ref{Abb:ProjAufSp3} and \ref{Abb:ProjAufSp2}.

A point $\v p$ with coordinates $p_1, p_2$ in $\koord K$ is called \emph{positive} if $p_1>0$ and $p_2 >0$, and it is called \emph{non-negative} if $p_1 \geq 0$ and $p_2 \geq 0$. We denote this by $\v p>\v 0$ and $\v p \geq \v 0$, respectively.

A symmetric ringnet $\affr M$ is called \emph{positive} or \emph{non-negative} if all points of $H(\affr M)$ are positive or non-negative, respectively. This is denoted by $\affr M> \v 0$ or $\affr M \geq \v 0$. 

\begin{DEF}\label{DEF:VergleichNetze}
\textbf{\emph{(Comparison of two symmetric ringnets)}}\\
Let $\affr N_1$ and $\affr N_2$ be two symmetric ringnets with the same topological structure and coordinate system $\koord K$. We call
$\affr N_1 > \affr N_2$ if $H(\affr N_1) > H(\affr N_2)$ holds pointwise in $\koord K$.
The relations $\geq$, $<$, $\leq$ and $=$ between two symmetric ringnets can be defined similarly.
\end{DEF}

In order to characterize the size of a ringnet in $\koord K$, we introduce a norm.

\begin{DEF}
\textbf{\emph{(Norm, $\mathrm{MAX}$, $\mathrm{MIN}$)}}\\
Let $\affr N$ be a symmetric ringnet with coordinate system $\koord K$. Let $\v p$ be a vertex with coordinates $p_1$ and $p_2$. Then we define
\[\Vert \v p \Vert :=\max\{|p_1 |, | p_2 | \}\;,\]
\[\mathrm{MAX}(\affr N) := \Vert \affr N \Vert := \max_{\v q \in H(\affr N)} \Vert \v q \Vert\]
and
\[\mathrm{MIN}(\affr N) := \min_{\v q \in H(\affr N)} \Vert \v q \Vert\;.\]
\end{DEF}

\begin{LEMMA}\label{LEMMA:VergleichbareEigenschaftenZweierSymmetrischenNetzen}
\textbf{\emph{(Comparable properties between two symmetric ring\-nets)}}\\
Let $\affr N_1$ and  $\affr N_2$ be symmetric ringnets with 
segment  angles $\varphi_1\in (0, \pi)$ and \mbox{$\varphi_2 \in [\varphi_1, \pi]$}, respectively, such that 
$H(\affr N_1)$ lies in $C(\affr N_1)$. If we rotate $\affr N_1$ by $(\varphi_2-\varphi_1)/4$ degrees counterclockwise, we obtain the following statements with respect to the coordinate system of $\affr N_2$ under the assumption $\affr N_1 \ge \affr N_2$:
\begin{enumerate}
\item[(a)] $\M M_n \affr N_1 \ge \M M_n \affr N_2$\;.
\item[(b)] If $\v 0 < (\affr N_1)_{0..\omega} \neq (\affr N_2)_{0..\omega}$,
then, for any $j\ge 0$ there exists an $\alpha \in (0, 1)$ such that
\[\alpha  (\M M_n^k \affr N_1)_{0..j} > (\M M_n^k \affr N_2)_{0..j}\]
for all sufficiently large $k$.
\end{enumerate}
\end{LEMMA}
\begin{proof}
{
Because $\affr N_i$ is symmetric, $H(\M M_n \affr N_i)$ can be generated from $H(\affr N_i) \cup \{\v 0\}$ by calculating midpoints and by projecting points into the spokes $S_0(\affr N_i)$ and $S_{0.5}(\affr N_i)$. Averaging preserves inequalities and the same holds for ortho\-gonal projections into the spokes as we see for two corresponding vertices
\[\v p_i = [x_{i}\; y_{i}]^\mathrm{t} = \v p_{kj}^0 (\affr N_i), \quad i=1, 2\;.\]
Let $\v q_i$ be their projections into $S_0(\affr N_i)$ and let $\v r_i$ be their projections into $S_{0.5}(\affr N_i)$, see Figure~\ref{Abb:ProjAufSp3}. 
Since $\v p_1 \in C(\affr N_1)$, we obtain
\[
\v p_1 \geq \v p_2 \quad \Rightarrow \quad \v q_1 \geq \v q_2 \;\; \mbox{and} \;\; \v r_1 \geq \v r_2\;,
\]
by which we derive Statement~(a).
\FigureCenter{h}{Abb:ProjAufSp3}
  {\includegraphics{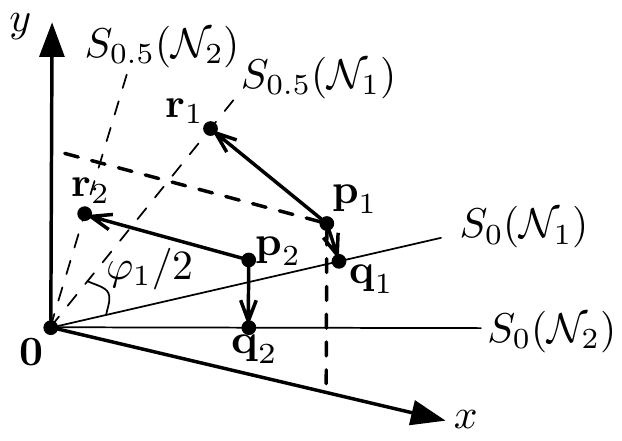}}
  {Projection into the spokes.}

Moreover, since $\varphi_1<\pi$, we get
\[\begin{array}{lll}
\v 0 < \v p_1 \geq \v p_2 \;\; \mbox{and} \;\; x_{1} > x_{2} & \Rightarrow & \v q_1 > \v q_2\;,\\
\v 0 < \v p_1 \geq \v p_2 \;\; \mbox{and} \;\; y_{1} > y_{2} & \Rightarrow & \v r_1 > \v r_2\;,
\end{array}
\] see Figure~\ref{Abb:ProjAufSp3}.
We have assumed $\affr N_1 \ge \affr N_2$ and $\v 0 < \v p_1 \ge \v p_2 \ne \v p_1$ for some points $\v p_i = \v p_{l_1 l_2}^0((\affr N_i)_{0..\omega}), i=1, 2$. According to 
Lemma~\ref{Lemma:BeeinflussDerRingeNachEinerUnterteilung}, the points $\v p_i$ influence all points in $(M_n^k(\affr N_i))_{0..\omega}$ for sufficiently large $k$ and for $i=1, 2$, respectively. 
Hence, $\v 0 < \v s_1 > \v s_2$ for $\v s_i = \v p_{11}^0(M_n^k(\affr N_i))$ and since $\v s_i$ influences all points in $(\M M_n^k \affr N_i)_{0..j}$ for all sufficiently large(r) $k$, we have derived Statement~(b). 
}
\end{proof}
                                       
\begin{LEMMA}\label{LEMMA:MIN-MAX-Beziehung}
\textbf{\emph{($\mathrm{MIN}$-$\mathrm{MAX}$ relation)}}\\
Let 
$\affr M^{(k)} := (\M M_n^k \affr M)_{0..\omega}$
be a subdivided grid mesh with segment angle in $(0, \pi)$. Then there exists a positive $\nu$ such that for all $k \geq 0$
\[\mathrm{MIN}(\affr M^{(k)}) / \mathrm{MAX}(\affr M^{(k)}) \geq \nu\;.\]
\end{LEMMA}
\begin{proof}
{
Let $\v p$ be a maximum vertex in $H(\affr M^{(k)})$, i.\,e.,  $\Vert \v p \Vert = \mathrm{MAX}(\affr M^{(k)})$. 

First, we apply 
Lemma~\ref{Lemma:BeeinflussDerRingeNachEinerUnterteilung} to obtain an $r \in \mathbb N$ such that $\v p$ influences every vertex $\v q$ in $H(\affr M^{(k+r)})$. The refinement and averaging operators $R$ and $A$ preserve symmetry. Therefore, any vertex $\v q$ in $H(\affr M^{(k+r)})$ can be computed by calculating midpoints of vertices or by projecting vertices orthogonally into the spokes $S_{0}$ and $S_{0.5}$. According to Lemma~\ref{LEMMA:PosInvariance}, the first half segment lies in the cone $C = C(\affr M)$. Since $C$ is non-negative, it follows for the midpoint of any two points $\v a, \v b \in C$ that
\begin{equation}\label{EQ:PunkteBez1}
\left\Vert \frac{1}{2} \v a + \frac{1}{2} \v b \right\Vert \geq \frac{1}{2} \left\Vert \v a \right\Vert\;.
\end{equation} 
Furthermore, if $\v b$ is the projection of $\v a$ into $S_{0.5}$ and if $\v c$ is the intersection point of $S_{0}$ and the line $\v a \v b$, we have  
\begin{equation}\label{EQ:PunkteBez2}
\Vert \v b\Vert = \frac{\Vert \v b\Vert}{\Vert \v c\Vert} \Vert \v c\Vert = \cos (\varphi/2) \; \Vert \v c\Vert \geq \cos (\varphi/2) \; \Vert \v a\Vert\;,
\end{equation}
as illustrated in Figure~\ref{Abb:ProjAufSp2}.
\FigureCenter{h}{Abb:ProjAufSp2}
  {\includegraphics{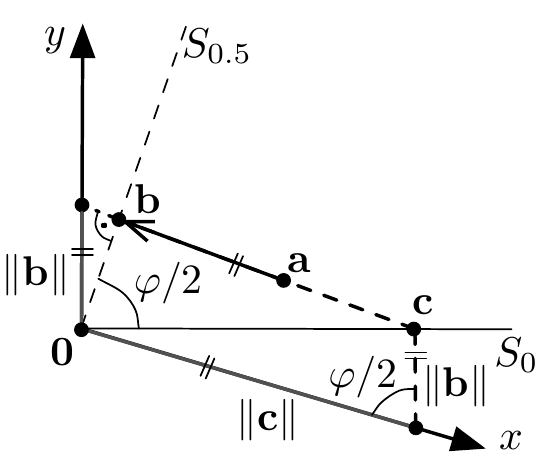}}
  {Projection into the spoke $S_{0.5}$.}

Note that in deriving Equation~\ref{EQ:PunkteBez1}, we used that the coordinate axes are perpendicular to the spokes $S_{0}$ and $S_{0.5}$. Because of symmetry, the same holds if $\v b$ is the projection of $\v a$ into $S_0$.
Altogether we see that there exists a constant $\mu> 0$ depending only on $n$ and $\varphi$ such that 
\[\mathrm{MIN}(\affr M^{(k+r)}) \geq \mu \; \mathrm{MAX}(\affr M^{(k)})\;.\]

Second, we retrace the above argument with Equation~\ref{EQ:PunkteBez1} by
\begin{eqnarray*}
\max\{\Vert \v a \Vert, \Vert \v b \Vert\} &\geq& \frac{1}{2} \Vert \v a \Vert + \frac{1}{2} \Vert \v b \Vert\\
&\geq& \left\Vert \frac{1}{2} \v a + \frac{1}{2} \v b \right\Vert
\end{eqnarray*}
and Equation~\ref{EQ:PunkteBez2} by $\Vert \v a \Vert \geq \Vert \v b \Vert$. 
This then shows
\[\mathrm{MAX}(\affr M^{(k)})) \geq \mathrm{MAX}(\affr M^{(k+r)}))\]
and thus
\[\mathrm{MIN}(\affr M^{(k+r)}) \geq \mu \; \mathrm{MAX}(\affr M^{(k+r)})\;.\]

Because $\mu$ and $r$ do not depend on $k$, Lemma~\ref{LEMMA:MIN-MAX-Beziehung} holds for
\[\nu = \min\left( \{\mu\} \cup \{ \mathrm{MIN}(\affr M^{(k)})  / \mathrm{MAX}(\affr M^{(k)}) \;|\; k=0, 1, \ldots, r-1\}\right)\;.\]
}
\end{proof}

\section{Subdividing grid meshes}\label{SECTION:KonvergenzDerFolge1}
In this section we show, that a sequence of subdivided grid meshes converges to an eigenvector of $M_n$, where we use the results of Section~\ref{SECTION:ComparingRingnets} to prove that the sequence cannot have two distinct accumulation points. Let $\affr N$ be a grid mesh with segment angle $\varphi \in (0, \pi)$ and $j+1$ rings, where $j \ge \omega=\left\lfloor \frac{n-1}{2} \right\rfloor$. Let 
$\affr N^{(k)} = (\M M_n^k \affr N)_{0..j}$ and
$\affr N_k = \affr N^{(k)} / \Vert \affr N^{(k)} \Vert$ for $k\ge 0$.
In this section, we show that the sequence of the normalized ringnets $\affr M_k=\affr N_{0..\omega} / \Vert \affr N_{0..\omega} \Vert$ converges (Theorem~\ref{SATZ:KonvMk}) to a symmetric positive eigennet of $M_n$ restricted to core meshes.

Because of Lemma~\ref{LEMMA:VergleichbareEigenschaftenZweierSymmetrischenNetzen}\,(b), $\affr M_k$ is positive for all sufficiently large $k$ and we can apply Lemma~\ref{LEMMA:MIN-MAX-Beziehung} showing that $\mathrm{MIN}(\affr M_k)$ is greater than a positive constant for all sufficiently large $k$. Hence, and with the aid of Lemma~\ref{LEMMA:PosInvariance}, we obtain
\begin{LEMMA}\label{LEMMA:K>0}
\textbf{\emph{($\affr K > \v 0$)}}\\
Every accumulation point $\affr K$ of $(\affr M_k)$ is positive and its first half segment lies in the cone spanned by the spokes $S_0$ and $S_{0.5}$ of $\affr M$,
\[H(\affr K) \subset C(\affr M)\;.\]
\end{LEMMA}

Let $\M M$ be the operator $\M M_n$ restricted to $j$-nets. It can be represented by a quadratic matrix with $\affr N^{(k)} = \M M^k \affr N$ because of Lemma~\ref{Lemma:BeeinflussDerRingeNachEinerUnterteilung}.
Let $\lambda_1, \ldots, \lambda_p$ be the eigenvalues of $\M M$ and let $\v v_i^0, \ldots, \v v_i^{\alpha_i}$ be the eigenvectors and the generalized eigenvectors associated with $\lambda_i$, i.\,e.,
\[\M M \v v^0_i = \lambda_i \v v^0_i \quad \mbox{and} \quad \M M \v v^l_i = \lambda_i \v v^l_i + \v v^{l-1}_i, \; l = 1,\ldots,\alpha_i\;.\]
With respect to these vectors, $\affr N$ has a unique decomposition
\[\affr N = \sum_{i=1}^p \sum_{l=0}^{\sigma_i} \nu_{i,l} \v v^{l}_i\]
with $\sigma_i\leq\alpha_i$ and $\nu_{i,\sigma_i}\neq 0$.

It can be readily checked that
\[\M M^k \v v_i^l = \binom{k}{l} \lambda_i^{k-l} \v v_i^0   + \v u_k\]
with $\Vert \v u_k\Vert = \mathrm{o}(|\binom{k}{l} \lambda_i^{k-l}|)$.
Assuming $|\lambda_1|=\cdots=|\lambda_s|$ and $\sigma_1 = \cdots = \sigma_s\ge 0$ and assuming for $i>s$ that $\sigma_i<\sigma_1$ if $|\lambda_i| = |\lambda_1|$, or $\sigma_i=-1$ if $|\lambda_i| > |\lambda_1|$, we get the following lemma.
\begin{LEMMA}\label{LEMMA:DarstellungHP1}
\textbf{\emph{(Expression of accumulation points)}}\\
The limit of any convergent subsequence of $(\affr N_k)$, i.\,e., any accumulation point of $(\affr N_k)$, lies in 
\[\v V := \mathrm{span} \{\v v_1^0, \ldots, \v v_s^0\}\;.\]
If $\lambda_1 = \cdots = \lambda_s>0$, then any accumulation point of $(\affr N_k)$ is an eigenvector of $\M M$ with eigenvalue $\lambda_1$.
\end{LEMMA}

\begin{LEMMA}\label{LEMMA:HP1}
\textbf{\emph{(Accumulation point $\v 1$)}}\\
Here, we multiply and raise vectors to a power coordinate wise.
The sequence $(\v r^k)$ in $\mathbb C^l$ with $\v r = [r_1 \; \ldots \; r_l]^\mathrm{t}$ and $|r_1|=\cdots = |r_l|=1$ has the accumulation point $\v 1:=[1\;\ldots\;1]^\mathrm{t}$.
\end{LEMMA}
\begin{proof}
{
The sequence $(\v r^k)$ is bounded and therefore has an accumulation point $\v h$. Hence, for all $\epsilon > 0$, there are integers $l_1$ and $l_2$ with arbitrarily large differences $l_1-l_2>0$ such that
\[\Vert \v r^{l_1} - \v h\Vert < \epsilon \quad \mbox{and} \quad \Vert \v r^{l_2} - \v h\Vert < \epsilon\;,\]
where $\Vert \cdot \Vert := \Vert \cdot \Vert_{\infty}$. Because $|r_i^{l_2}| = 1$, we have
\begin{eqnarray*}
\Vert \v r^{l_1-l_2} - \v 1\Vert &=& \Vert (\v r^{l_1-l_2} - \v 1) \v r^{l_2}\Vert\\
&=& \Vert \v r^{l_1} - \v r^{l_2}\Vert\\
&\leq& \Vert \v r^{l_1} - \v h\Vert + \Vert \v r^{l_2} - \v h\Vert < 2 \epsilon\;,
\end{eqnarray*}
which concludes the proof.
}
\end{proof}

\begin{SATZ}\label{SATZ:KonvMk}
\textbf{\emph{(Convergence of $(\affr M_k)$)}}\\
The sequence $(\affr M_k)$ converges.
\end{SATZ}
\begin{proof}
{
Since $(\affr M_k)$ is bounded, it suffices to show that $(\affr M_k)$ has only one accumulation point. Let $\affr K$ and $\affr L$ be two accumulation points of $(\affr M_k)$. Using Lemma~\ref{LEMMA:DarstellungHP1} with $j=\omega$, we can write $\affr K$ and $\affr L$ as
\[\affr K = \sum_{i=1}^s \nu_i \v v_i^0 \quad \mbox{and}\;\quad\affr L = \sum_{i=1}^s \mu_i \v v_i^0\;.\]
Let $\lambda = |\lambda_1|$, $r_i = \lambda_i/\lambda$ and $\v r^k = [r_1^k \; \ldots \; r_s^k]^\mathrm{t}$. Then,
\[\frac{1}{\lambda^k} \M M^k \affr K = [\nu_1 \v v_1^0 \; \ldots \; \nu_s \v v_s^0] \v r^k\quad
\mbox{and}\quad
\frac{1}{\lambda^k} \M M^k \affr L = [\mu_1 \v v_1^0 \; \ldots \; \mu_s \v v_s^0] \v r^k\;.\] 
Due to Lemma~\ref{LEMMA:HP1}, there is a sequence $(k_n)$ with $\v r^{k_n} \to \v 1$ as $n\to \infty$. Therefore,
\begin{equation}\label{EQ:K_L}
\frac{1}{\lambda^{k_n}} \M M^{k_n} \affr K \to \affr K\quad \mbox{and}\quad\frac{1}{\lambda^{k_n}} \M M^{k_n} \affr L \to \affr L\;.
\end{equation}

From Lemma~\ref{LEMMA:K>0}, it follows that $\affr K > \v 0$ and $\affr L > \v 0$. 
Therefore, there is a largest sequence $(a_k) \subset \mathbb R_{>0}$ and a smallest sequence $(b_k) \subset \mathbb R_{>0}$ such that
\begin{equation}\label{EQ:akbk}
a_k \; \left(\frac{1}{\lambda}\right)^k \M M^k \affr K \leq \left(\frac{1}{\lambda}\right)^k \M M^k \affr L \leq b_k \; \left(\frac{1}{\lambda}\right)^k \M M^k \affr K 
\end{equation}
for all $k \geq 0$.
Since $(a_k)$ is maximal and $(b_k)$ is minimal, $(a_k)$ is monotonically increasing and $(b_k)$ monotonically decreasing because of Lemma~\ref{LEMMA:VergleichbareEigenschaftenZweierSymmetrischenNetzen}\,(a). 
Since $(a_k)$ is bounded by $b_0$ and $(b_k)$ by $a_0$, both sequences converge. Let $a=\lim a_k$ and $b=\lim b_k$. 
Then we infer from \ref{EQ:K_L} and \ref{EQ:akbk} that
\[a \affr K \leq \affr L \leq b \affr K\;,\]
and again using \ref{EQ:akbk} that $a\leq a_0 \leq b_0 \leq b$.
Consequently, the monotonicity implies $a_k=a$ and $b_k=b$ for all $k$.

Next, we assume $\affr L \ge a \affr K \neq \affr L$. Due to Lemma~\ref{LEMMA:K>0}, we can apply Lemma~\ref{LEMMA:VergleichbareEigenschaftenZweierSymmetrischenNetzen}\,(b) and get $\M M^k \affr L > a \M M^k \affr K$ for all sufficiently large $k$, which contradicts the maximality of $a_k=a$. Thus, $\affr L = a\affr K$.

Since $1 = \Vert \affr M_k \Vert = \Vert \affr K \Vert = \Vert \affr L \Vert = \Vert a \affr K \Vert = a \Vert \affr K \Vert$, we obtain $a=1$ and $\affr L = \affr K$, which concludes the proof.
}
\end{proof}

\begin{FOLGERUNG}\label{FOLGERUNG:KonvergenzMk}
\textbf{\emph{(Properties of the limit ringnet)}}\\
Let $\affr M_{\infty}$ be the limit ringnet of $(\affr M_k)$.
\begin{enumerate}
\item[(a)] $\affr M_{\infty}$ is an eigenvector of $\M M$ with eigenvalue $\lambda = \Vert \M M \affr M_{\infty} \Vert >0$.
\item[(b)] $\affr M_{\infty} > \v 0$.
\item[(c)] $\affr M_{\infty}$ has norm $1$.
\item[(d)]  Since $M$ preserves symmetry, $\affr M_{\infty}$ is symmetric with the same segment angle as $\affr M$.
\end{enumerate}
\end{FOLGERUNG}
\begin{proof}
Because 
\begin{eqnarray*}
\frac{\M M \affr M_{\infty}}{\Vert \M M \affr M_{\infty} \Vert} &=& \frac{\M M \; \lim_{k \to \infty} \frac{\M M^k \affr M}{\Vert \M M^k \affr M \Vert}}{\left\Vert \M M \; \lim_{k \to \infty} \frac{\M M^k \affr M}{\Vert \M M^k \affr M \Vert} \right\Vert}
= \lim_{k \to \infty} \frac{\frac{\M M^{k+1} \affr M}{\Vert \M M^k \affr M \Vert}}{\left\Vert \frac{\M M^{k+1} \affr M}{\Vert \M M^k \affr M \Vert} \right\Vert}\\
&=& \lim_{k \to \infty} \frac{\M M^{k+1} \affr M}{\Vert \M M^{k+1} \affr M \Vert}
= \affr M_{\infty}\;,
\end{eqnarray*}
we get Statement~(a), and by Lemma~\ref{LEMMA:K>0} Statement~(b), and Statement~(c) has been shown in the proof for Theorem~\ref{SATZ:KonvMk}.
\end{proof}

\section{Subdominant eigenvalues of the midpoint operator restricted to core meshes}\label{SECTION:Subdominance}
Let $\affr M$ be a grid mesh with $\omega+1$ rings, frequency $f$ and segment angle $\varphi = 2f \pi/m \in (0, \pi)$.
Let $\M M$ be the matrix of $\M M_n$ restricted to core meshes and let $\affr M^{(k)} = \M M^k \affr M$ for $k\ge 0$. Due to Theorem~\ref{SATZ:KonvMk} and Corollary~\ref{FOLGERUNG:KonvergenzMk}, the sequence $\affr M_k := \affr M^{(k)} / \Vert \affr M^{(k)}\Vert$ converges to a symmetric eigennet $\affr M_\infty$ with segment angle $\varphi$ and some certain positive eigenvalue $\lambda$. In this section, we show that $\lambda$ is the dominant eigenvalue for the eigenspace of frequency $f$ and that $\lambda$ is smaller for higher frequencies. Moreover, we derive that for frequency $1$, $\lambda$ is a double subdominant eigenvalue of $M$.

\begin{SATZ}\label{SATZ:EigenschaftLimesnetzes}
\textbf{\emph{(Maximum property of $\affr M_\infty$)}}\\
Any rotation symmetric eigennet $\affr N$ of $\M M$ with segment angle $\varphi_{\affr N} \in [\varphi, \pi]$ has an eigenvalue $\mu$ with $|\mu|<\lambda$ or is a multiple of  $\affr M_\infty$.
\end{SATZ}
\begin{proof}
{
It suffices to consider the case that $\affr N$ is not a multiple of  $\affr M_\infty$.

First, we assume that $\affr N$ is reflection symmetric. As in Lemma~\ref{LEMMA:VergleichbareEigenschaftenZweierSymmetrischenNetzen}, we rotate $\affr M_\infty$ by $(\varphi_{\affr N}-\varphi)/4$ degrees counterclockwise. We use the coordinate system of $\affr N$ and normalize $\affr N$ such that
\[\affr N \leq \affr M_\infty \quad \mbox{and} \quad \affr N \nless \affr M_\infty\;.\]
From Lemma~\ref{LEMMA:VergleichbareEigenschaftenZweierSymmetrischenNetzen}\,(b) it follows that for sufficiently large $k$
\[\M M^{2k} \affr N < \M M^{2k} \affr M_\infty \quad \mbox{and} \quad \mu^{2k}\affr N < \lambda^{2k}\affr M_\infty\;.\]
From Lemma~\ref{LEMMA:EigenwertVSymmEigennetz} it follows that $\mu^2$ is positive, which yields $|\mu|<\lambda$.

Second, we assume that $\affr N$ is not reflection symmetric and $\mu\ne 0$. Using the scalar product $\langle \affr N, \affr M_\infty\rangle := \affr N^\mathrm{t} \overline{\affr M}_\infty$, we define
\[\affr N_1 := \left\{ \begin{array}{rll} \affr N &\mbox{if}& \langle \affr N, \affr M_\infty\rangle = 0 \\ \frac{\I}{\langle \affr N, \affr M_\infty\rangle}\affr N &\mbox{if}& \langle \affr N, \affr M_\infty\rangle \neq 0 \end{array} \right.\]
and obtain the symmetric ringnet
\[\affr N_s = \affr N_1 + \widetilde{\overline{\affr N_1}}\]
with segment angle $\varphi_{\affr N}$.
Since $\langle \widetilde{\overline{\affr N}}, \widetilde{\overline{\affr M}}_\infty \rangle = \overline{\langle \affr N, \affr M_\infty\rangle}$, we have $\langle \affr N_s, \affr M_\infty\rangle = 0$ and, hence, $\affr N_s \neq \beta \affr M_\infty$ for all $\beta \in \mathbb C\backslash \{0\}$.

If $\affr N_s = \v 0$, then $\I \affr N_1 = \widetilde{\overline{\I \affr N_1}}$ is a symmetric ringnet with eigenvalue $\mu$ and, as shown above, $|\mu| < \lambda$.

If $\affr N_s \neq \v 0$, we rotate $\affr M_\infty$ by $(\varphi_{\affr N}-\varphi)/4$ degrees counterclockwise as in Lemma~\ref{LEMMA:VergleichbareEigenschaftenZweierSymmetrischenNetzen}. We use the coordinate system of $\affr N_s$ and normalize $\affr N_s$ such that
\begin{equation}\label{EQ:NsK}
\affr N_s \leq \affr M_\infty \quad \mbox{and} \quad \affr N_s \nless \affr M_\infty\;.
\end{equation}
Because of Lemma~\ref{LEMMA:VergleichbareEigenschaftenZweierSymmetrischenNetzen}\,(b),  there exists an $\alpha\in (0,1)$ such that for sufficiently large $k$
\[\mu^k \affr N_1 + \overline{\mu}^k \widetilde{\overline{\affr N_1}} =  \M M^k \affr N_s < \alpha \M M^k \affr M_\infty = \alpha \lambda^k \affr M_\infty\]
and
\[\left(\frac{\mu}{|\mu|}\right)^k \affr N_1 + \left(\frac{\overline{\mu}}{|\mu|}\right)^k \widetilde{\overline{\affr N_1}} < \alpha \left(\frac{\lambda}{|\mu|}\right)^k \affr M_\infty\;.\]
According to Lemma~\ref{LEMMA:HP1}, there is a sequence $(k_l)_{l\in\mathbb N}$ with $(\mu/|\mu|)^{k_l} \to 1$. Consequently $|\mu|<\lambda$ because otherwise
\[\lim_{l\to \infty} \left(\frac{1}{|\mu|}\right)^{k_l} \M M^{k_l} \affr N_s = \affr N_s \leq \alpha \affr M_\infty < \affr M_\infty\;,\]
which contradicts \ref{EQ:NsK}.
}
\end{proof}

The subdivision operator $\M M$ can be represented by a block-cyclic matrix. Consequently, there is a basis of rotation symmetric (generalized) eigenvectors with frequencies \mbox{$f \in \{0,\ldots, m-1\}$}, see \cite[Section 3]{PR98}, \cite[Section 4.2]{ZS2001} and \cite[Section 3.5 on Pages 47-58]{Chen2005}.
Since the conjugate eigenvectors have the frequencies $m-f$, we obtain
\begin{FOLGERUNG}\label{FOLGERTUNG:LambdaPhi}
\textbf{\emph{(Dominant eigenvalues)}}\\
The eigenvalue $\lambda$ of $\affr M_\infty$ is the dominant eigenvalue associated with frequencies $f$ and $m-f$. (Since $\affr M_\infty$ and $\lambda$  depend only on $\varphi$ and the fixed $n$, we denote $\lambda$ by $\lambda_\varphi$.)
\end{FOLGERUNG}

For $m=4$, the eigenvalues $\lambda_\varphi$ are well-known. Around a regular vertex, the midpoint subdivision surface consists of four polynomial patches and subdividing these means to map any polynomial $p(z)$ onto $p(z/2)$. Hence the eigenfunctions of $\M M$ are monomials and the eigenvalues are powers of $1/2$. Since the least monomials with frequency $0$ and $2$ are quadratic, besides the constant $1$, we obtain that $\mu_0 = 1/4$ is the subdominant eigenvalue for the eigenspace with frequency $0$ and that $\lambda_\pi = 1/4$. 
Furthermore, Theorem~\ref{SATZ:EigenschaftLimesnetzes} implies

\begin{equation}\label{EQ:Lambda}
\lambda_\alpha > \lambda_\beta > \lambda_\pi = \mu_0 = 1/4 \quad \mbox{for} \quad 0 < \alpha < \beta < \pi\;.
\end{equation}
Consequently $\lambda_{2\pi/m}$ is the subdominant eigenvalue of $\M M$.

\begin{SATZ}\label{SATZ:DominanteEigenwerte}
\textbf{\emph{(The subdominant eigenvalue of $M_n$ restricted to \mbox{core meshes)}}}\\
The eigenvalue $\lambda_{2\pi/m}$ of $M_n$ restricted to core meshes is subdominant and has algebraic and geometric multiplicity $2$.
\end{SATZ}
\begin{proof}
{
It remains to show that $\lambda=\lambda_{2\pi/m}$ has algebraic multiplicity $2$ or that there is no generalized eigenvector associated with $\lambda$. We assume that $\affr H$ is a generalized eigenvector with frequency $1$ and associated with $\lambda$, i.\,e., $\M M \affr H = \lambda \affr H + \affr M_\infty$. Then
\[\affr H_s = \frac{1}{2}(\affr H + \widetilde{\overline{\affr H}})\]
is a symmetric generalized eigenvector of $\M M$ with frequency $1$ because 
\begin{eqnarray*}
\M M \affr H_s &=& (\M M \affr H + \M M  \widetilde{\overline{\affr H}})/2
= (\lambda \affr H + \affr M_\infty + \widetilde{\overline{\lambda \affr H + \affr M_\infty}})/2\\
&=& (\lambda \affr H + \affr M_\infty + \lambda \widetilde{\overline{\affr H}} + \affr M_\infty)/2
= \lambda \affr H_s + \affr M_\infty\;.
\end{eqnarray*}
Therefore, it can be presupposed that $\affr H$ is symmetric. Because $\affr M_\infty> \v 0$, there exists an $\alpha>0$ with $\alpha \affr M_\infty \geq \affr H$. The ringnet $\alpha \affr M_\infty - \affr H$ is symmetric of frequency $1$ and is non-negative. From Lemma~\ref{LEMMA:VergleichbareEigenschaftenZweierSymmetrischenNetzen}\,(a) it follows that
\begin{equation}\label{EQ:Hauptvektor}
\M M^k (\alpha \affr M_\infty - \affr H) \geq \v 0 
\end{equation}
for all $k\geq 0$.
Furthermore, we have
\begin{eqnarray*}
\M M^k (\alpha \affr M_\infty - \affr H) &=& \alpha \M M^k \affr M_\infty - \M M^k \affr H\\
&=& \alpha \lambda^k \affr M_\infty - \left(\lambda^k \affr H + k \lambda^{k-1} \affr M_\infty\right)\\ 
&=& \lambda^{k} \left( \left(\alpha-\frac{k}{\lambda}\right) \affr M_\infty -  \affr H \right)\\
&<&\v 0
\end{eqnarray*}
for all sufficiently large $k$, which contradicts \ref{EQ:Hauptvektor}.
Therefore, the algebraic multiplicity of $\lambda$ is $2$.
}
\end{proof}

\section{Subdivision of larger ringnets}\label{SECTION:LargerRingnets}
As shown in Section~\ref{SECTION:Subdominance}, subdividing a core grid mesh with frequency $1$ leads to the subdominant eigenvectors of $M_n$ restricted to core meshes.
In this section, we generalize this result to arbitrary large grid meshes. 

We consider a $j$-net $\affr N$ with $j\geq \omega+1$, where the control points are arranged so that 
\[\affr N = \left[ \begin{array}{l} \affr N_{0..\omega} \\ \affr N_b \\ \affr N_c \\ \affr N_{\omega+2..j}\end{array} \right]\;,\]
where $\affr N_c$ consists of the corners of the ring $\affr N_{\omega+1}$, i.\,e.,
$\affr N_c = \{\v p^{i}_{\omega+1,\omega+1} \;|\; i \in \mathbb Z_m\}$ if $\affr N$ is primal or
$\affr N_c = \{\v p^{i}_{\omega+2, \omega+2} \;|\; i \in \mathbb Z_m\}$ if $\affr N$ is dual, and $\affr N_b$ consists of the other points in $\affr N_{\omega+1}$. 
With this arrangement of the control points and according to Lemma~\ref{Lemma:BeeinflussDerRingeNachEinerUnterteilung}, the matrix of $\M M_n$ restricted to $j$-nets has the lower triangular form
\[\M S = \left[ \begin{array}{cccccc} \M A\\ * & \M B\\ * & * & \M C \\  * & * & * & 0\\
\vdots & & & \ddots & \ddots\\
  * & \hdots & \hdots & \hdots & * & 0
\end{array} \right]\;,\]
where
\[(\M M_n \affr N)_{0..\omega} = \M A \affr N_{0..\omega}\;,\]
\[\left[ \begin{array}{l} (\M M_n \affr N)_{0..\omega} \\ (\M M_n\affr N)_b \end{array} \right] =  \left[ \begin{array}{cc} \M A \\ * & \M B \end{array} \right]\left[ \begin{array}{l} \affr N_{0..\omega} \\ \affr N_b \end{array} \right]\;,\]
\[(\M M_n \affr N)_{0..\omega+1} = \left[ \begin{array}{ccc} \M A\\ * & \M B\\ * & * & \M C\end{array} \right] \affr N_{0..\omega+1}\;.\]
Hence, the eigenvalues of $\M S$ are zero or are the eigenvalues of the blocks $\M A$, $\M B$ and $\M C$.

\begin{LEMMA}\label{LEMMA:ZeilensummenVonC}
\textbf{\emph{(Spectral radii of $\M B$ and $\M C$)}}\\
The spectral radii $\rho_{\M B}$ and $\rho_{\M C}$ of $\M B$ and $\M C$ satisfy
\[\rho_{\M B} \leq 2^{-n} \quad \mbox{and} \quad  \rho_{\M C}\leq 4^{-n}\;.\]
\end{LEMMA}
\begin{proof}
{
Since $\M C$ is non-negative, we get \cite[Corollary 6.1.5 on Page 346]{HJ85}
\[\rho_{\M C} \le \Vert \M C \Vert_\infty = \Vert \M C \v 1\Vert_\infty, \quad\mbox{where} \quad \v 1 := [1\;\ldots\;1]^{\mathrm{t}}\;.\]
The vector $\M C \v 1$ represents the corners of the $(\omega+1)$-ring $\affr N_{\omega+1}'$, where $\affr N_{0..\omega} = \v 0$, $\affr N_{\omega+2..j} = \v 0$, $\affr N_b = \v 0$ and $\affr N_c = \v 1$. One can easily verify that the vertices of  these corners are  $\le 4^{-n}$, which concludes the proof of the second statement. The first statement can be proven similarly.
}
\end{proof}

According to Theorem~\ref{SATZ:DominanteEigenwerte}, Inequality~\ref{EQ:Lambda} and Lemma~\ref{LEMMA:ZeilensummenVonC}, we get
\begin{SATZ}\label{SATZ:DominanteEigenwerteLargerNets}
\textbf{\emph{(The subdominant eigenvalue of $M_n$ restricted to \mbox{$j$-nets} for $j\ge \omega+1$)}}\\
The subdominant eigenvalue of $M_n$ restricted to $j$-nets with $j\ge \omega+1$ is $\lambda_{2\pi/m}$ and $\lambda_{2\pi/m}$ has algebraic and geometric multiplicity $2$.
\end{SATZ}

\begin{SATZ}\label{SATZ:CharMesh}
\textbf{\emph{(The characteristic mesh of $\M M_n$)}}\\
Let $\affr N$ be a grid mesh with frequency $1$ and with $\rho+1$ rings, where $\rho = \left\lceil \frac{3}{2}n - \frac{3}{2} \right\rceil$, see Equation~\ref{EQ:RHO}. Let $\M S_\rho$ be the matrix of $\M M_n$ restricted to $\rho$-nets. Let 
\[n_k := \Vert \M S_\rho^k \affr N \Vert \quad \mbox{and} \quad
\affr N_k := \M S_\rho^k \affr N/n_k\;.\]
Then for $n\ge 2$, $(\affr N_k)$ converges  to a positive subdominant eigennet, called the \emph{characteristic mesh} of $\M M_n$.
\end{SATZ}
\begin{proof}
{Let
$m_k := \Vert (\M S_\rho^k \affr N)_{0..\omega} \Vert$
and let
$\affr M_k := (\M S_\rho^k \affr N)_{0..\omega}/m_k$.
According to Theorem~\ref{SATZ:KonvMk} and Corollary~\ref{FOLGERUNG:KonvergenzMk}\,(b), the sequence $(\affr M_k)$ converges to a positive eigennet $\affr M_\infty$ with eigenvalue $\lambda=\lambda_{2\pi/m}$.
The dominant and subdominant eigenvalues $1$ and $\lambda$ of $\M S_\rho$ are positive, see Corollary~\ref{FOLGERUNG:KonvergenzMk}\,(a). Furthermore, since
\begin{equation}\label{thetalambda}
n_k \ge m_k \in \Theta(\lambda^k)\;,
\end{equation}
we obtain by Lemma~\ref{LEMMA:DarstellungHP1} that $(\affr N_k)$ converges to an eigenvector $\affr N_\infty$ with eigenvalue $1$ or $\lambda$.
If $(\affr N_\infty)_{0..\omega}$ were zero, then $\M S_\rho \affr N_\infty = \mu N_\infty$ with \mbox{$|\mu|\le 2^{-n}<\lambda$} according to Lemma~\ref{LEMMA:ZeilensummenVonC} and Inequality~\ref{EQ:Lambda}~, which contradicts~\ref{thetalambda}.
Consequently $(\affr N_\infty)_{0..\omega}$ is a multiple of $\affr M_\infty$ and is positive and $\lambda$ is its eigenvalue. 
Using Lemma~\ref{LEMMA:VergleichbareEigenschaftenZweierSymmetrischenNetzen}\,(b), we have $\lambda^k \affr N_\infty = M_n^k \affr N_\infty > \v 0$ for all sufficiently large $k$ and hence, $\affr N_\infty > \v 0$.
}
\end{proof}

\section{The characteristic map}\label{SECTION:CharacteristicMap}
Let $\affr C$ be the characteristic mesh of valence $m$ of the midpoint scheme $\M M_n$. It defines the control mesh of a characteristic map, which is a spline surface ring consisting of $m$ segments with $3$ subsegments of $\lfloor n/2 \rfloor^2$ polynomial patches, as illustrated in Figure~\ref{Abb:SplineSegment}.

\begin{SATZ}\label{SATZ:RegInjOfCharMap}
\textbf{\emph{(Characteristic map)}}\\
The characteristic map of $\M M_n$ is regular and injective for $n\ge 2$.
\end{SATZ}

\FigureCenter{!h}{Abb:SplineSegment}
  {\includegraphics{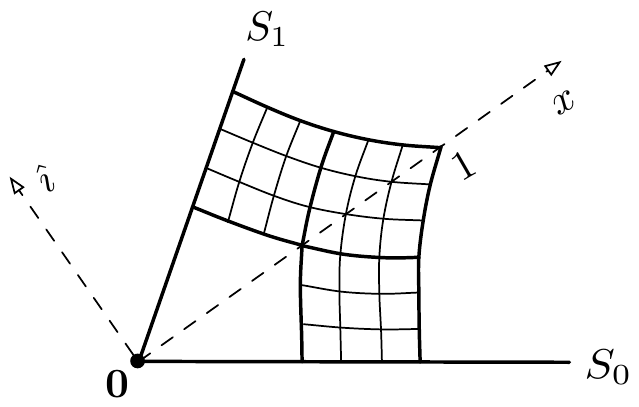}}
  {One segment and its $3\cdot 9$ polynomial patches for $n=6, 7$.}

\begin{proof}
{
Let $\Omega = [0, 2]^2 \backslash [0, 1)^2$ and let $\v c(u,v): \Omega \to \mathbb C$ be the first segment of the characteristic map rotated such that $\v c(1,1)>0$. According to \cite[Theorem 5.25 on Page 107]{PR2008}, the characteristic map is regular and injective if all partial derivatives $\v c_v(u, v)$ lie in
\[Q = \{x+\I y | x, y>0\}\;.\]
To show this, we consider a  rotated grid mesh $\affr M$ such that the subdivided and normalized meshes $\affr M_k = (\M M_n^k \affr M)_{0..\rho}/\Vert (\M M_n^k \affr M)_{0..\rho}\Vert$ converge to the characteristic mesh $\affr C$.

Let $\affr E_k$ and $\affr E$ be the sets of all edge directions $\v p_{i, j+1}^0 - \v p_{i, j}^0$ in the first segment of $\affr M_k$ and $\affr C$, respectively. These and other edge directions control the  partial derivatives $\v c_v$ of the first segment. Furthermore, we add both $\v u_1$ and $\I \v u_0$ to $\affr E_k$ and $\affr E$, where $\v u_0$ and $\v u_1$ are the edge directions of the spokes $S_0$ and $S_1$, respectively. Refining and averaging a mesh also means to average its edge directions by the masks shown in Figure~\ref{Abb:MaskEdges}.

\FigureCenter{h}{Abb:MaskEdges}
  {\includegraphics{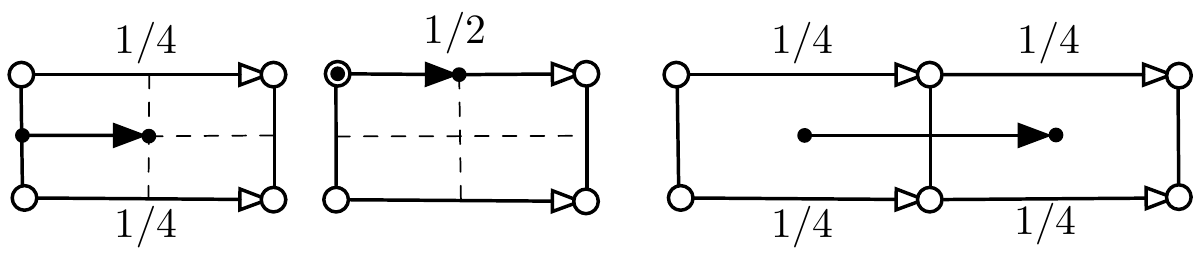}}
  {Masks for the edge directions.}

In particular, the directions in $\affr E_k$ are either, due to symmetry, parallel to $\v u_1$ and $\I \v u_0$, or obtained by iteratively averaging the directions in $\affr E_{k-1}$ and multiplying these by positive numbers because of the normalization. Thus, we know that $\affr E_k$ lies in the cone spanned by $\affr E_{k-1}$, i.\,e., the cone spanned by $\{\v u_1, \I \v u_0\}$, and therefore, both $\affr E_k$ and $\affr E$ lie in the larger set $Q \cup \{0\}$.

Moreover, since $\affr C_{0..1}$ is symmetric and has at most one zero control point, at least one of its edge directions is non-zero. Subdividing $\affr C$, we can see that every element of $\affr E$ is a linear combination of $\affr E$ with non-negative weights and a positive weight for the non-zero element in the $1$-ringnet. Hence $\affr E$ has no zero elements and lies in $Q$.

A partial derivative $\v c_v(\v u), \v u \in [0, 2]\times[1, 2]$ (or $\v u \in [1, 2]\times[0, 2]$) is a convex combination of directions in $\affr E$ or in $\affr E$ reflected at $S_1$ (or at $\I S_0$), where a reflected direction has a smaller weight than its unreflected counterpart. Hence, $\v c_v(\Omega) \subset Q$ as claimed.
}
\end{proof}

Theorems~\ref{SATZ:DominanteEigenwerteLargerNets} and \ref{SATZ:RegInjOfCharMap} show that Reif's $C^1$-criterion \cite[Theorem 3.6]{Reif95} is satisfied. 
So finally, we have established
\begin{SATZ}\label{SATZ:C1Mn}
\textbf{\emph{($C^1$-property of $\M M_n$)}}\\
The midpoint scheme $\M M_n$ is a $C^1$-sub\-division algorithm for all degrees $n\geq 2$ and for all valencies $m\geq 3$.
\end{SATZ}

\addcontentsline{toc}{section}{References}
\bibliographystyle{alpha}
\bibliography{./arXiv_AnalyzingMidpointSubdivision}







\end{document}